\documentclass[a4paper,12pt]{article}
\usepackage[english]{babel}
\usepackage[utf8]{inputenc}
\usepackage{amsthm,amsmath,amsfonts,amssymb,amscd,cases}
\usepackage{mathtools}
\usepackage{dsfont}
\usepackage[T1]{fontenc}

\newcommand{\R}{{\mathord{\mathbb R}}}

\newcommand{\N}{{\mathord{\mathbb N}}}

\def\inf{{\rm inf}\,}

\DeclareMathOperator{\supp}{supp}

\theoremstyle{plain}
\newtheorem{lemma}{Lemma}[section]
\newtheorem{theorem}[lemma]{Theorem}

\theoremstyle{definition}

\theoremstyle{remark}
\newtheorem{remark}[lemma]{Remark}

\newtheorem*{remark*}{Remark}

\title{Quantum Systems at The Brink. Existence and Decay Rates of Bound States at Thresholds; Atoms}

\author{Dirk Hundertmark, Michal Jex, Markus Lange}
\begin{document}

\maketitle

\begin{abstract}
It is well known that $N$-electron atoms undergoes unbinding for a critical charge of the nucleus $Z_c$, i.e. the atom has eigenstates for the case $Z> Z_c$ and it has no bound states for $Z<Z_c$. In the present paper we derive upper bound for the bound state for the case $Z=Z_c$ under the assumption $Z_c<N-K$ where $K$ is the number of electrons to be removed for atom to be stable for $Z=Z_c$ without any change in the ground state energy. We show that the eigenvector decays faster as $\exp\left(-C\sum\sqrt{|x|_{k}}\right)$ where we sum K largest values of $|x_j|$, $j\in\{1,\ldots,N\}$. Our method do not require Born-Oppenheimer approximation.
\end{abstract}

\section{Introduction}
We consider an atom with $N$ electrons. Its energy is described  by the Hamiltonian 
\begin{align*}
H_Z^{(N)} = \sum_{j=1}^N\left( -\Delta_j - \frac{Z}{|x_j|}\right) + \sum_{j\neq k} \frac{1}{|x_j -x_k|}\,
\end{align*}
where $Z$ is in principle an arbitrary positive number. We choose units in such a way that $\frac{\hbar^2}{2m}=1$ and $\frac{e^2}{4\pi\epsilon_0}=1$. The domain of our Hamiltonian is the antisymmetric subspace of $L^2(\R^{3N})$, i.e. all square integrable functions which satisfy
\begin{align*}
\psi(\ldots,x_i,\ldots,x_j,\ldots)=-\psi(\ldots,x_j,\ldots,x_i,\ldots)\,,\quad i\neq j\,.
\end{align*}
 For the simplicity we do not take into the account the spin degrees of freedom. However one can add them without any difficulty. Also our proof does not rely explicitly on fermionic statistics and is applicable also for bosons or distinguishable particles.\\ 
 
We are interested in the behaviour of the ground state for the case that its eigenvalue is at the edge of the essential spectrum. One expects the existence of a critical coupling $Z_c$ for which all the bound states below the essential spectrum disappear. This is supported by classical results by Zhislin \cite{Zhi60} which says that there are bound states for the case $Z>N-1$ and Lieb \cite{Lie84} which showed nonexistence of bound states for $N\geq 2Z+1$. The nonexistence result was improved by Nam \cite{Nam12} to $N\geq 1.22Z+3Z^{\frac 1 3}$. The existence of bound state for the critical coupling was done in \cite{BelFraLieSei14} under the additional assumption that $Z<N-K$ where $K$ is the number of removed electrons during the transition. Gridnev \cite{Gri12} showed the existence result for the case $Z_c\in(N-2,N-1)$. His result is applicable for systems without Born-Oppenheimer approximation. Our theorem in fact also provides an alternative proof of this existence result.\\

The existence and absence for the case $N=2$ was studied extensively by Hoffmann-Ostenhof, Hoffmann-Ostenhof and Simon in the 80'. In \cite{HofOstHofOstSim83} they showed that for the distinguishable particles, i.e. electrons with spin for example, the ground state for critical coupling $Z_c\sim0.91$ exists. However in \cite{Hof84} it was proved that the situation changes for fermions without spin, i.e. there is no ground state for critical coupling $Z_c=1$. To be more precise they showed that there is no antisymmetric function depending only on $|x_1|, |x_2|$ and $x_1\cdot x_2$ which can be a ground state of the system.\\

In the present paper we show that the bound state for the critical coupling behaves as
\begin{align*}
\psi(x_1,\ldots,x_N)\leq\exp\left(-C\sum_{k}\sqrt{|x_{k}|}\right)
\end{align*}
where we sum over $K$ biggest values of $|x_{k}|$. In fact we can even improve this result and show that the eigenfunction decays exponentially unless the remaining $N-K$ electrons are localized within a small ball around the nucleus. We give a more rigorous description of this claim later on. To show our result we apply the method developed in \cite{HunJexLan19-Helium}.\\

Our paper is organized as follows. The introduction is concluded by an overview of the method used to prove the main theorem. In Section 2 we state our main result. In the last section we give the proof which is split into two parts. The first part summarizes auxiliary results needed in the second part where we prove our theorem.

\subsection{Introduction of Our Method}
\label{method}
In this section we describe step by step our method without technical details. The advantage of our method is that we do not require a gap between the eigenvalue and the threshold of the essential spectrum which is necessary for other methods, e.g. Agmon method. The main ingredient of our method is to use the repulsive parts of potentials in the considered Hamiltonians to get some extra freedom and remove the necessity of a safety distance with respect to the bottom of the essential spectrum.\\
\\
\textbf{Starting point}\\
We consider a selfadjoint operator $H$ and a normalized eigenvector $\psi$ satisfying
\begin{equation*}
	H\psi=E\psi\,,
\end{equation*}
where $E$ is the corresponding eigenvalue below or at the threshold of the essential spectrum.\\

\textbf{1st step}\\
We introduce two functions $\chi_R$ as a cutoff function with a support outside a compact region and $\zeta$ as a sequence of function related to the decay which we want to show. We calculate
\begin{equation*}
	\mathrm{Re}\langle(\zeta\chi_R)^2\psi,H\psi\rangle=E\langle(\zeta\chi_R)^2\psi,\psi\rangle=E\|\zeta\chi_R\psi\|^2
\end{equation*}

\textbf{2nd step}\\
Next we use a variant of IMS formula \cite{CycFroKirSim87} to obtain
\begin{equation*}
\langle\zeta\chi_R\psi,H\zeta\chi_R\psi\rangle-\langle\psi,|\nabla\zeta\chi_R|^2\psi\rangle=E\|\zeta\chi_R\psi\|^2
\end{equation*}
At this point we split $|\nabla\zeta\chi_R|^2$ into two terms. One of them is compactly supported, we denote it by $G$, and the other one is the rest $B$.\\

\textbf{3rd step}\\
Rearranging terms and approximating we get
\begin{equation*}
	\left\langle\zeta\chi_R\psi,\left(H-E-B\right)\zeta\chi_R\psi\right\rangle\leq \|G\psi\|^2\leq K
\end{equation*}

\textbf{Final step}\\
We show that the expression $H-E-B$ is positive. We note that here the repulsive part of the potential in $H$ may come in handy. If the expression is positive we can conclude that $\zeta\psi$ has a bounded norm. Hence $\psi$ has the expected decay behavior. We remark at this point that if $\inf\sigma_{ess}H-E$ is positive we obtain analog of Agmon method \cite{Agm82}.

%%%%%%%%%%%%%%%%%%%%%%%%%%%%%%%%%%%%%%%%%%%%%%%%%%%%%%%%%%%%%%%%%%%%%%%%%%%%%
%%%%%%%%%%%%%%%%%%%%%%%%%%%%%%%%%%%%%%%%%%%%%%%%%%%%%%%%%%%%%%%%%%%%%%%%%%%%%%

\section{Main Result} \label{sec:NAtom}
In the following we consider an atom with $N$ electrons. 
We make the standard assumptions that the nucleus is infinitely heavy and at the origin, i.e. 
Born-Oppenheimer approximation. 
We denote by $x_i$ the position operator for $N$ electrons, $i \in \{1,\ldots,N\}$. 

We define the Hamiltonian of this system by
\begin{equation}\label{eq:NAtommHamiltonian}
	H_Z^{(N)} = \sum_{j=1}^N -\Delta_j - \frac{Z}{|x_j|} + \sum_{j\neq k} \frac{1}{|x_j -x_k|}
\end{equation}
where $-\Delta_j$ is the kinetic energy of the $j$-th electron. 
% Need to specify units, to properly!!!!
It is well-defined and selfadjoint on $D(H_Z^{(N)}) \subseteq L^2_a(\R^{3N})$. Note that we consider electrons to be fermions. However our approach would work also for bosonic or distinguishable particles. Actually one part of the proof would be simpler for distinguishable particles.

\textbf{Goal:}
We are interested in the decay rate of normalized  eigenfunctions $\psi_Z$ of $H_Z^{(N)}$ for the 
critical case $Z = Z_c$.

It is well-known that $E_Z^{(N)}$ is non-increasing, concave function of $Z$ with the property
\begin{align*}
	E_Z^{(N)} \leq E_Z^{(N-1)} \leq E_Z^{(N-2)} \leq \ldots
\end{align*}

We denote the ground state energy of the $N$-electron Hamiltonian by 
\begin{align*}
	E_Z^{(N)} = \inf \sigma\left(H_Z^{(N)}\right)\,,
\end{align*}
where $\sigma(\cdot)$ denotes the spectrum. For subcritical values of $Z$ the existence of corresponding ground states follows from Zhislin's theorem and the HVZ theorem. For the critical value $Z_c$ existence of a ground state was shown in \cite{BelFraLieSei14}. 
\begin{theorem}
Let $Z_c$ be the maximum of the set
\begin{align*}
 \{Z>0|E_Z^{(N)}=E_Z^{(N-1)}\,\mathrm{and}\,E_{Z_n}^{(N)}<E_{Z_n}^{(N-1)}\,\mathrm{for\,some}\,Z_n\rightarrow Z\}
\end{align*}
and assume that $Z_c<N-K$ where $K$ is the largest integer such that $E_{Z_c}^{(N)}=E_{Z_c}^{(N-K)}$. Then $H^{(N)}_{Z_c}$ has a ground state eigenfunction  $0\not\equiv\psi_{Z_c}\in L^2_a(\R^{3N})$.
\end{theorem}

\begin{remark}
There is an alternative way to obtain the existence of such an eigenfunction employing our main result. It is based on tightness arguments as described in \cite{HunLee12} which give necessary conditions for a weakly converging sequence to be a strongly convergent one, namely
\begin{align*}
\lim_{R\rightarrow\infty}\limsup_{n\rightarrow\infty}\int_{|x|>R}|\psi_n(x)|^2\textrm{d}x=0\,,\\
 \lim_{L\rightarrow\infty}\limsup_{n\rightarrow\infty}\int_{|k|>L}|\hat\psi_n(k)|^2\textrm{d}k=0\,,
\end{align*}
The essential ingredient in the proof of existence is then that our decay estimate works uniformly for all $Z>Z_c$. For more details we refer the reader to \cite{HunJexLan19-Helium}.
\end{remark}

In order to formulate and prove our main result we need to construct a splitting of the electron space into several regions. We begin by splitting the electrons into two groups. For that we  introduce the function $|x|_k:\R^{3N}\rightarrow\R^+_0$  for $k\in \{1,\ldots,N\}$. This function gives the $k$-th smallest value out of $|x_j|$ for every $j\in \{1,\ldots,N\}$ including degeneracy, i.e.
\begin{equation}\label{eq:ordering}
|x|_1\leq|x|_2\leq\ldots\leq|x|_N\,,\quad\forall x\in\R^{3N}\,.
\end{equation}
It is obvious that there always exists a permutation $\pi\in\mathfrak S_N$ such that $|x_{\pi(k)}|=|x|_k$. This permutation might not be unique for cases when one of the inequalities in \eqref{eq:ordering} is not sharp. We denote the set of all possible permutations for each $x$ by $\mathfrak S_N^A$. 
For a given permutation $\pi\in\mathfrak S_N^A$ and given point $x\in\R^{3N}$ we split the electron coordinates into the following two groups \\

Inner coordinates: $X_{I,K}^\pi=\left\{x_{\pi(k)}\big|k\in\{1,\ldots N-K\}\right\}$,\\
Outer coordinates: $X_{O,K}^\pi=\left\{x_{\pi(k)}\big|k\in\{N+1-K,\ldots N\}\right\}$.\\

where we omit $x$ in the notation. Now we are ready to state the main result

\begin{theorem}\label{thm:FallOff}
Let $H_Z^{(N)}$ be given by Eq.~\eqref{eq:NAtommHamiltonian} and let $\psi_{Z}\in L_a^2(\R^{3N})$ be normalized function such that 
$H_Z^{(N)}\psi_{Z} =E_Z^{(N)}  \psi_{Z}$. Then 
\begin{equation*}
	e^G \psi_Z \in L_a^2(\R^{3N})\,,
\end{equation*}
where
\begin{equation}\label{eq:RegionDependingFallOff}
	G := \begin{cases} 
		\sum_{m=N+1-K}^NC_m\sqrt{|x|_m}\,,
		\quad &x\in\R^{3N}\,\,s.t.\,\,\frac\delta 2|x|_{N-K+1}> |x|_{N-K}\\
		\sum_{m=N+1-K}^NK_m|x|_m\,, &\mathrm{otherwise} %if |x|_\infty \geq (1-\delta_B) |x_1 - x_2| 
	\end{cases}
\end{equation}
for given small enough $\delta>0$.% and constants $C_m, K_m$.
\end{theorem}

\begin{remark}
We note that the condition $\delta|x|_{N-K+1}\geq |x|_{N-K}$ implies that such a relation holds for every combination of outer and inner coordinates, namely
\begin{equation*}
\forall\pi\in\mathfrak S_N^A\,\forall y\in X_{I,K}^\pi, z\in X_{O,K}^\pi:\delta|z|\geq |y|\,.
\end{equation*}
\end{remark}

\begin{remark}
Our result remains valid also if we consider finite mass nucleus. For such a situation one needs to make several straighforward modifications as in \cite[Appendix B]{HunJexLan19-Helium}. However we need to add an assumption that the mass of the nucleus is at least as large as the sum of all masses of electrons. This is satisfied in physically relevant situations.
\end{remark}

\section{Proof of Theorem~\ref{thm:FallOff}}
Before stating the proof we prepare several auxiliary results. We also introduce various notations used in the proof. 

\subsection{Preliminary Estimates}
For the purpose of our proof we need to separate $\R^{3N}$ into several subsets which satisfy given symmetry requirements, 
i.e. the eigenfunction multiplied by the symmetric function is still antisymmetric. In order to achieve this we choose these subsets to be symmetric. We work with the subsets of
\begin{equation*}
\Omega=\{x\in\R^{3N}\,:\,|x|_1<\ldots<|x|_N\}\subseteq\R^{3N}\,.
\end{equation*}
It is easy to see that $\Omega$ differs from $\R^{3N}$ by the measure zero set
\begin{equation*}
\partial\Omega=\{x\in\R^{3N}\,:\,\exists i\neq j\,\mathrm{s.t.}\,\, |x_i|=|x_j|\}\,.
\end{equation*}

We start by estimating $U=\sum_{j=1}^N- \frac{Z}{|x_j|} + \sum_{j\neq k} \frac{1}{|x_j -x_k|}$ in $\Omega$ for the case of $K$ outer coordinates. We first introduce notation $\tilde x_j$ s.t. $|\tilde x_j|=|x|_j$ which is unique and well defined on $\Omega$. In the following we denote by $U_M:=\sum_{j=1}^M- \frac{Z}{|x|_j} + \sum_{j\neq k}^M \frac{1}{|\tilde x_j -\tilde x_k|}$ the potential corresponding to inner $M$ particles. The estimation of $U$ is performed in an iterative way. We start with 
\begin{equation*}
\Omega= \begin{cases} 
		A_1: |x|_1>\delta|x|_N\\
		A_1^c: |x|_1<\delta|x|_N\,.
	\end{cases}
\end{equation*}
We obtain
\begin{equation*}
U|_{A_1}\geq\left( U_{N-K-1}-\left(1+\frac{K}{\delta}\right)\frac{Z}{|x|_N}\right)\Bigg|_{A_1}
\end{equation*}
where we used $\frac{1}{|x_j-x_k|}>0$ and $-\frac{1}{|x|_j}>-\frac{1}{\delta|x|_N}$ due to the fact $|x|_1<|x|_j$, $j>1$. For the other case we have
\begin{equation*}
U|_{A_1^c}\geq\left( U_{N-1}+\sum_{j=2}^{N-1}\frac{1}{|\tilde x_j-\tilde x_N|}+\left(\frac{1}{1+\delta}-Z\right)\frac{1}{|x|_N}\right)\Bigg|_{A_1^c}
\end{equation*}
where we used $\frac{1}{|\tilde x_1-\tilde x_N|}>\frac{1}{|\tilde x_1|+|\tilde x_N|}>\frac{1}{(1+\delta)|\tilde x_N|}$. Provided that $K>1$ we split $A_1^c$ as
\begin{equation*}
A_1^c= \begin{cases} 
		A_2: |x|_1<\delta|x|_N\,\,\mathrm{and}\,\,|x|_2>\delta|x|_N\\
		A_2^c: |x|_1<\delta|x|_N\,\,\mathrm{and}\,\,|x|_2<\delta|x|_N\,
	\end{cases}
\end{equation*}
which implies
\begin{align*}
U|_{A_2}&\geq\left( U_{N-K-1}-\left(1+\frac{K}{\delta}\right)\frac{Z}{|x|_N}\right)\Bigg|_{A_2}\\
U|_{A_2^c}&\geq\left( U_{N-1}+\sum_{j=3}^{N-1}\frac{1}{|\tilde x_j-\tilde x_N|}+\left(\frac{2}{1+\delta}-Z\right)\frac{1}{|x|_N}\right)\Bigg|_{A_2^c}\,.
\end{align*}
where we used $\frac{1}{|\tilde x_2-\tilde x_N|}>\frac{1}{|\tilde x_2|+|\tilde x_N|}>\frac{1}{(1+\delta)|\tilde x_N|}$. We can repeat this process $N-K$ times. The last step yields
\begin{equation*}
A_{N-K-1}^c= \begin{cases} 
		A_{N-K}: |x|_1<\delta|x|_N,\,\ldots,\,|x|_{N-K-1}<\delta|x|_N\,\,\mathrm{and}\,\,|x|_{N-K}>\delta|x|_N\\
		A_{N-K}^c: |x|_1<\delta|x|_N,\,\ldots,\,|x|_{N-K}<\delta|x|_N\,
	\end{cases}
\end{equation*}
and
\begin{align*}
U|_{A_{N-K}}&\geq\left( U_{N-K-1}-\left(1+\frac{K}{\delta}\right)\frac{Z}{|x|_N}\right)\Bigg|_{A_{N-K}}\\
U|_{A_{N-K}^c}&\geq\left( U_{N-1}+\sum_{j=N-K+1}^{N-1}\frac{1}{|\tilde x_j-\tilde x_N|}+\left(\frac{N-K}{1+\delta}-Z\right)\frac{1}{|x|_N}\right)\Bigg|_{A_{N-K}^c}\,.
\end{align*}
Now we can repeat this process with $U_{N-1}$ with number of outer coordinates equal to $K-1$ in the region $A_{N-K}^c$.
This is summarized in the following lemma.

\begin{lemma}\label{lem:EstOp}
Let $U=- \frac{Z}{|x_j|} + \sum_{j\neq k} \frac{1}{|x_j -x_k|}$. Then
\begin{equation*}
U|_\Omega\geq \begin{cases} 
U_{N-K-1}-\frac{Z}{\delta|x|_{N-K+1}}-\sum_{j=N-K+1}^{N}\frac{Z}{|x|_j}\quad&|x|_{N-K}>\delta|x|_{N-K+1}\\
U_{N-K}+\sum_{j=N-K+1}^{N}\left(\frac{N-K}{1+\delta}-Z\right)\frac{1}{|x|_j}\quad&|x|_{N-K}<\delta|x|_{N-K+1}\,.
	\end{cases}
\end{equation*}
\end{lemma}

Next we define a positive smooth function $\chi_{R}$ that will play the role of the cutoff function required in the proof of the main theorem. We define
\begin{equation}\label{eq:cutofffunc}
\chi_{R}:= \begin{cases} 
		1\,,\quad x\in\{x\in\R^{3N}\,:\,\forall m\geq N-K+1 \,, |x|_m\geq R\}\\
		0\,,\quad x\notin\{x\in\R^{3N}\,:\,\forall m\geq N-K+1 \,, |x|_m\geq R/2\}\\
		\in[0,1]\,,\quad\mathrm{otherwise}
	\end{cases} \,.
\end{equation}
The second function, which we need, is
\begin{equation}\label{eq:cutofffunc2}
\chi_{0,\delta}:= \begin{cases} 
		1\,,\quad x\in\Omega\,\,\textrm{s.t.}\,\,|x|_{N-K}<\frac \delta 2|x|_{N-K+1}\\
		0\,,\quad x\in\Omega\,\,\textrm{s.t.}\,\,|x|_{N-K}>\delta|x|_{N-K+1}
	\end{cases} \,.
\end{equation}
We also require $\chi_{0,\delta}$ to be homogeneous of order $0$. We note that such a choice is possible because our condition are homogeneous. Additionally, we introduce its complement
\begin{equation*}\label{eq:cutofffunc2}
\chi_{0,\delta}^\perp=\sqrt{1-\chi_{0,\delta}^2}\,\quad\forall x\in\Omega\,.
\end{equation*}
and the notation
\begin{align*}
\chi_{R,\delta}&=\chi_R\chi_{0,\delta}\\
\chi_{R,\delta}^\perp&=\chi_R\chi_{0,\delta}^\perp\,.
\end{align*}		
We also introduce the following function which works as a upper bound for the function defined in Eq.~\eqref{eq:RegionDependingFallOff}
\begin{equation}\label{eq:RegionDependingFallOffEst}
	F := \sum_{m=N+1-K}^N C_m\sqrt{|x|_m}+K_m|x|_m\chi^\perp_{0,2\delta} %if |x|_\infty \geq (1-\delta_B) |x_1 - x_2| 
\end{equation}
Directly from the definition we see that $G\leq F$ for all $x\in\supp\chi_R$. Last but not least we need the following estimate 
\begin{lemma}\label{lem:EstGrad}
Let $F$ be the function defined in Eq. \eqref{eq:RegionDependingFallOffEst}. Then
\begin{equation*}
|\nabla F|^2 \leq \begin{cases} 
\sum_{m=N-K+1}^Nd_m^2+\frac{e_m}{\sqrt{|x|_m}},		\quad& x\in\mathrm{supp}\chi_{0,2\delta}^\perp\\
\sum_{m=N-K+1}^Nc_m^2\frac{1}{|x|_m},	\quad&\textrm{otherwise} 
	\end{cases}
\end{equation*}
where $c_m,d_m,e_m>0$.
\end{lemma}

\begin{proof}
We start by calculating
\begin{equation*}
|\nabla |x|_m|^2=\left|\sum_{j=1}^N\partial_j|x|_m\right|^2=\left|\partial_k|x_k|\right|^2=1
\end{equation*}
%where $|x|_m$ is the $m$-th smallest value out of $|x_j|$, $\forall j\in\{1,\ldots,N\}$ counting multiplicity.
For $x\notin\mathrm{supp}\chi_{0,2\delta}^\perp$ we can express $|\nabla F|^2$ as
\begin{equation*}
|\nabla F|^2 =\left|\nabla \sum_{m=N+1-K}^NC_m\sqrt{|x|_m}\right|^2=\left|\sum_{m=N+1-K}^NC_m\nabla\sqrt{|x|_m}\right|^2=\sum_{m=N+1-K}^N\left|C_m\right|^2\frac{1}{4|x|_m}\,.
\end{equation*}
In the other case, i.e. $x\in\mathrm{supp}\chi_{0,2\delta}^\perp$, we write
\begin{align*}
|\nabla F|^2 =&
		\left|\nabla \sum_{m=N+1-K}^NC_m\sqrt{|x|_m}+K_m|x|_m\chi^\perp_{0,2\delta}\right|^2\\
		=&\left|\sum_{m=N+1-K}^N\nabla \left(C_m\sqrt{|x|_m}+K_m|x|_m\chi^\perp_{0,2\delta}\right)\right|^2\\
		\leq&\sum_{m=N+1-K}^N\left|\frac{C_m}{\sqrt{|x|_m}}+K_m\chi^\perp_{0,2\delta}+K_m|x|_m\nabla\chi^\perp_{0,2\delta}\right|^2\,.
\end{align*}
We need to check that $|x|_m\nabla\chi^\perp_{0,2\delta}$ is uniformly bounded. We know that $\chi^\perp_{0,2\delta}$ is homogeneous of order $0$. This implies that $\nabla\chi^\perp_{0,2\delta}$ is homogeneous of order $-1$. A direct consequence of this is that $|x|_m\nabla\chi^\perp_{0,2\delta}$ is homogeneous function of the order $0$.
\end{proof}

\subsection{Proof of Main Theorem}
Now we are ready to proof the theorem. For the convenience of the reader we highlight the steps of our method described in Subsection~\ref{method}. We remark that in the course of the proof we use a form variant of IMS formula where we relax  requirement for cut-off functions to  piecewise-$C^1$ functions. The details are provided in the appendix.\\

\begin{proof}[Proof of Theorem~\ref{thm:FallOff}]
Let $H_Z^{(N)}$ be given by Eq.~\eqref{eq:NAtommHamiltonian} and let $\psi \in L_a^2(\R^{3N})$ 
be an eigenfunction such that $H_Z^{(N)}\psi_Z = E_Z^{(N)} \psi_Z$, i.e. 
$\psi$ is a ground state for $H_Z^{(N)}$. 
We define $\xi_{R,\delta}:=\chi_{R,\delta}\exp\left(\frac{F}{1+\epsilon F}\right)$ and  $\xi^\perp_{R,\delta}:=\chi^\perp_{R,\delta}\exp\left(\frac{F}{1+\epsilon F}\right)$ for $\epsilon>0$.\\
\textbf{Step 1}:  Starting from $H_Z^{(N)}\psi_Z=E_Z^{(N)}   \psi_Z $ we get
\begin{align*}
\mathrm{Re}\left\langle(\xi_{R,\delta}^2+(\xi_{R,\delta}^\perp)^2)\psi, H_Z^{(N)}\, \psi\right \rangle= E_Z^{(N)}( \|\xi_{R,\delta} \psi\|^2+\| \xi_{R,\delta}^\perp\psi\|^2)
\end{align*}			
\textbf{Step 2}: Using a generalized variant of the IMS localization formula \cite{CycFroKirSim87} we obtain			
			\begin{equation}\label{eq:ineqstep2}
			\begin{split}
			\langle \xi_{R,\delta} \psi,H_Z^{(N)}  \xi_{R,\delta}\psi\rangle - \langle  \psi, |\nabla\xi_{R,\delta}|^2\psi \rangle 
= E_Z^{(N)}\langle \xi_{R,\delta} \psi, \xi_{R,\delta}\psi\rangle\\
			\langle \xi^\perp_{R,\delta} \psi,H_Z^{(N)}  \xi^\perp_{R,\delta}\psi\rangle - \langle  \psi, |\nabla\xi^\perp_{R,\delta}|^2\psi \rangle
= E_Z^{(N)}\langle \xi^\perp_{R,\delta} \psi, \xi^\perp_{R,\delta}\psi\rangle
\end{split}
\end{equation}
We rewrite the localization error in the following way
\begin{align*}
|\nabla\xi_{R,\delta}|^2&=
\left|( \nabla\chi_{R})\chi_{0,\delta}+\chi_R(\nabla\chi_{0,\delta})+\chi_{R,\delta}\frac{\nabla F}{(1+\epsilon F)^2}\right|^2\exp\left(\frac{2F}{1+\epsilon F}\right)\\
&\leq|\chi_R|^2\left|\nabla\chi_{0,\delta}+(\chi_{0,\delta}\nabla F)\right|^2\exp\left(\frac{2F}{1+\epsilon F}\right)+K\,,
\end{align*}
and
\begin{align*}
|\nabla\xi_{R,\delta}^\perp|^2&=
\left|( \nabla\chi_{R})\chi_{0,\delta}^\perp+\chi_R(\nabla\chi_{0,\delta}^\perp)+\chi_{R,\delta}\frac{\nabla F}{(1+\epsilon F)^2}\right|^2\exp\left(\frac{2F}{1+\epsilon F}\right)\\
&\leq|\chi_R|^2\left|\nabla\chi_{0,\delta}^\perp+(\chi_{0,\delta}^\perp\nabla F)\right|^2\exp\left(\frac{2F}{1+\epsilon F}\right)+K\,,
\end{align*}
where $K$ is a suitable uniform constant independent on $\epsilon$. Note that the inequality holds because $\nabla\chi_{R}$ is compactly supported.\\

\textbf{3rd step:} We rearrange the terms and calculate inequalities in three separate regions. This is equivalent to the calculation of quadratic forms in Eq.~\eqref{eq:ineqstep2} on 3 separate disjoint regions of $\Omega\subseteq\R^{3N}$:
\begin{align*}
S_1&:=\supp{\chi_{R}}\cap\{x\in\Omega:\chi_{0,\delta}(x)=1\}\,,\\
S_2&:=\supp{\chi_{R}}\cap\{x\in\Omega:\chi_{0,\delta}(x)=0\}\,,\\
S_3&:=\supp{\chi_{R}}\cap\{x\in\Omega:\chi_{0,\delta}(x)\in(0,1)\}\,.
\end{align*}
Hence we get
\begin{align*}
&\left\langle \xi_{R,\delta}\psi, (H_Z^{(N)}-E_Z^{(N)}) \xi_{R,\delta}\psi\right\rangle+\left\langle \xi_{R,\delta}^\perp \psi, (H_Z^{(N)}-E_Z^{(N)})  \xi_{R,\delta}^\perp\psi\right\rangle =\\
&\left\langle \xi_{R,\delta} \psi,(H_Z^{(N)}-E_Z^{(N)})  \xi_{R,\delta}\psi \right\rangle_{S_1}+\left\langle \xi_{R,\delta}^\perp \psi,(H_Z^{(N)}-E_Z^{(N)}) \xi_{R,\delta}^\perp\psi\right\rangle_{S_2}+\\
&\left\langle \xi_{R,\delta} \psi,(H_Z^{(N)}-E_Z^{(N)})  \xi_{R,\delta}\psi\right\rangle_{S_3}+\left\langle \xi_{R,\delta}^\perp \psi,(H_Z^{(N)}-E_Z^{(N)}) \xi_{R,\delta}^\perp\psi\right\rangle_{S_3}\,\\
\textrm{and}\\
&\langle  \psi, (|\nabla\xi_{R,\delta}|^2+ |\nabla\xi_{R,\delta}^\perp|^2 )\psi \rangle=\\
&\left \langle  \psi, |\nabla\xi_{R,\delta}|^2\psi \right\rangle_{S_1}+ \langle  \psi, |\nabla\xi_{R,\delta}^\perp|^2\psi \rangle_{S_2}+
 \langle  \psi, (|\nabla\xi_{R,\delta}|^2+ |\nabla\xi_{R,\delta}^\perp|^2 )\psi \rangle_{S_3}\,.
\end{align*}
Now we investigate following equalities:
\begin{align*}
&\left\langle \xi_{R,\delta} \psi,(H_Z^{(N)}-E_Z^{(N)})  \xi_{R,\delta}\psi\rangle_{S_1}= \langle  \psi, |\nabla\xi_{R,\delta}|^2\psi \right\rangle_{S_1}\\
&\left\langle \xi_{R,\delta}^\perp \psi,(H_Z^{(N)}-E_Z^{(N)}) \xi_{R,\delta}^\perp\psi\right\rangle_{S_2} = \langle  \psi, |\nabla\xi_{R,\delta}^\perp|^2\psi \rangle_{S_2}\\
&\left\langle \xi_{R,\delta} \psi,(H_Z^{(N)}-E_Z^{(N)})  \xi_{R,\delta}\psi\right\rangle_{S_3}+\left\langle \xi_{R,\delta}^\perp \psi,(H_Z^{(N)}-E_Z^{(N)}) \xi_{R,\delta}^\perp\psi\right\rangle_{S_3} \\&\quad = \langle  \psi, (|\nabla\xi_{R,\delta}|^2+ |\nabla\xi_{R,\delta}^\perp|^2 )\psi \rangle_{S_3}
\end{align*}
In the following we use several abbreviations
\begin{align*}
W_1&:=-Z\left(1+\frac 1 \delta\right)<0\,,\\
W_2&:=\frac{N-K}{1+\delta}-Z\,,\\
\mathfrak e_R&:=\chi_{R}\exp\left(\frac{F}{1+\epsilon F}\right)\,.
\end{align*}
We note that
\begin{align*}
\left \langle\varphi,\left(-\sum_{j=1}^N\Delta+U_M\right)\varphi\right\rangle\geq\left \langle\varphi,E^{(M)}_Z\varphi\right\rangle
\end{align*}
for all $\varphi\in\mathcal D(H_Z^{(N)} )$.
Using Lemma~\ref{lem:EstOp} and the fact that $E_Z^{(N-K)}\geq E_Z^{(N)}$ we obtain
\begin{align*}
&\left \langle\mathfrak e_R\psi,\left(\sum_{k=N+1-K}^N\frac{W_2}{|x|_{k}}\right)\mathfrak e_R\psi\right\rangle_{S_1}- \left\langle \mathfrak e_R\psi,\left|\nabla F\right|^2\mathfrak e_R\psi\right\rangle_{S_1}\leq K\\
&\left  \langle\mathfrak e_R\psi,\left(E_Z^{(N-K-1)}-E_Z^{(N)}+\sum_{k=N+1-K}^N\frac{W_1}{|x|_{k}}\right) \mathfrak e_R\psi\right\rangle_{S_2}- \left\langle\mathfrak e_R\psi,\left|\nabla F\right|^2\mathfrak e_R\psi\right\rangle_{S_2}\leq K\\
&\left \langle \mathfrak e_R\psi,\left(\sum_{k=N+1-K}^N\frac{W_2}{|x|_{k}}\right) \mathfrak e_R\psi\right\rangle_{S_3}\\
&-\left\langle \psi, \left(\left|\nabla\chi_{0,\delta}+(\chi_{0,\delta}\nabla F)\right|^2+\left|\nabla\chi_{0,\delta}^\perp+(\chi_{0,\delta}^\perp\nabla F)\right|^2\right)\mathfrak e_R^2\psi \right\rangle_{S_3}\leq K
\end{align*}
where we used that $\chi_{0,\delta}$ and $\chi_{0,\delta}^\perp$ sum up in squares to one and the fact that for sufficiently large $R$
\begin{align*}
\left(E_Z^{(N-K-1)}-E_Z^{(N)}+\sum_{k=N+1-K}^N\frac{W_1}{|x|_{k}}\right) \geq\left(\sum_{k=N+1-K}^N\frac{W_2}{|x|_{k}}\right)
\end{align*}
 holds.
Rearranging terms in the equation above we get
\begin{equation}\label{eq:finalIneq}
\begin{split}
&\left \langle \mathfrak e_R\psi,A_1 \mathfrak e_R\psi\right\rangle_{S_1}\leq K\\
&\left \langle \mathfrak e_R\psi,A_2 \mathfrak e_R\psi\right\rangle_{S_2}\leq K\\
&\left \langle \mathfrak e_R\psi,A_3 \mathfrak e_R\psi\right\rangle_{S_3}\leq K
\end{split}
\end{equation}
where
\begin{align*}
A_1&:=\sum_{k=N+1-K}^N\frac{W_2}{|x|_{k}}-|\nabla F|^2\\
A_2&:=E_Z^{(N-K-1)}-E_Z^{(N)}+\sum_{k=N+1-K}^N\frac{W_1}{|x|_{k}}-|\nabla F|^2\\
A_3&:=\sum_{k=N+1-K}^N\frac{W_2}{|x|_{k}}-\left|\nabla\chi_{0,\delta}+(\chi_{0,\delta}\nabla F)\right|^2-\left|\nabla\chi_{0,\delta}^\perp+(\chi_{0,\delta}^\perp\nabla F)\right|^2
\end{align*}
are restricted to appropriate regions.\\

\textbf{Final step:} We show that $A_i$ are positive. 
We start with $A_1$. Using Lemma~\ref{lem:EstGrad} with $N-K>Z$ and appropriately chosen constants $c_k$ we have
\begin{align*}
A_1\geq\sum_{k=N+1-K}^N\frac{W_2}{|x|_{k}}-\frac{c_k^2}{|x|_k}>0\,.
\end{align*}
Analogously for $A_2$ using Lemma~\ref{lem:EstGrad} we acquire
\begin{align*}
A_2\geq E_Z^{(N-K-1)}-E_Z^{(N)}+\sum_{k=N+1-K}^N\frac{W_1}{|x|_{k}}-d_k^2+\frac{e_k}{\sqrt{|x|_k}}>0\,
\end{align*}
where we used the assumption that $E_Z^{(N-K-1)}-E_Z^{(N)}>0$ and the fact that we take sufficiently large $R$. For $A_3$ we write
\begin{align*}
A_3\geq\sum_{k=N+1-K}^N\frac{W_2}{|x|_{k}}-2\left|\nabla\chi_{0,\delta}\right|^2-2\left|\nabla F\right|^2-2\left|\nabla\chi_{0,\delta}^\perp\right|^2-2\left|\nabla F\right|^2\,.
\end{align*}
One can check that $\left|\nabla\chi_{0,\delta}\right|\leq\frac{c}{|x|}$. This follows from $\chi_{0,\delta}$ being homogeneous function of degree $0$ and analogously for $\chi_{0,\delta}^\perp$. Using Lemma~\ref{lem:EstGrad} and the fact that $\mathrm{supp}(\chi_{0,\delta})\cap\mathrm{supp}(\chi^\perp_{0,2\delta})=\emptyset$ we obtain
\begin{align*}
A_3\geq\sum_{k=N+1-K}^N\frac{W_2}{|x|_{k}}-\frac{4c_k^2}{|x|_k}-\mathcal O\left(\frac{1}{|x|_k^{2}}\right)>0
\end{align*}
for a suitable choice of constants $c_m$ and large enough $R$. Recalling the definition of $\mathfrak e_R$ it remains to take the limit $\epsilon\rightarrow0$ in \eqref{eq:finalIneq} in order to acquire the desired result.
\end{proof}
%\hfill{$\square$}
%\hfill10$\box$

\begin{appendix}
\section{Form representation of IMS localization formula}
The well known IMS localization formula \cite{CycFroKirSim87} can be written as
\begin{align*}
\langle\psi,-\Delta\psi\rangle=\langle\xi\psi,-\Delta\xi\psi\rangle-\langle\psi,|\nabla\xi|^2\psi\rangle
\end{align*}
where $\psi\in H^2(\R^d)$ and $\xi\in C^\infty(\R^d)$. Very often the requirement on smoothness of $\xi$ is relaxed to $C^2(\R^d)$. We want to relax conditions even further. This is possible by restating the problem in quadratic form sense. We strongly believe that this reformulation should be known but we were not able to find it in the literature. Our proof heavily relies on weak formulation of Gauss Theorem \cite[Sec. A.6.8]{alt92}, explicitly

\begin{theorem}
Let $\Omega\subset\R^d$ be open, bounded and with Lipschitz boundary.
\begin{enumerate}
\item If $u\in W^{1,1}(\Omega)$ then for all $i=1,\ldots, n$
\begin{align*}
\int_{\Omega}\partial_i u \,\mathrm dx^d=\int_{\partial\Omega}u\nu_i  \mathrm dS^{d-1}
\end{align*}
holds where $\nu$ is other normal to $\partial\Omega$.
\item Let $1\leq p\leq\infty$.  If $u\in W^{1,p}(\Omega)$ and $v\in W^{1,q}(\Omega)$ with $\frac 1 p+\frac 1 q=1$ then for all $i=1,\ldots, n$
\begin{align*}
\int_{\Omega}(u\partial_i v+v\partial_i u) \mathrm dx^d=\int_{\partial\Omega}uv\nu_i  \mathrm dS^{d-1}
\end{align*}
holds where $\nu$ is outer normal to $\partial\Omega$.
\end{enumerate}
\end{theorem}

Using this theorem we can prove the following.

\begin{theorem}\label{gauss}
Let $\varphi\in H^1(\R^d)$, $\xi\in C(\R^d)$ bounded. Assume that there exist countably many open sets $\Omega_j$ with  Lipschitz boundary  such that $\R^{d}$ can be written as
\begin{align}
\label{union}
\R^{d}=\overline{\dot{\bigcup}_{j\in\N}\Omega_j}\,.
\end{align}
 and such that $\nabla\xi$ exists everywhere except on $\bigcup_{j\in\N}\partial\Omega_j$. Furthermore assume that $\nabla\xi$ is bounded. Then
\begin{align*}
\mathrm{Re}\langle\nabla(\xi^2\varphi),\nabla\varphi\rangle=\langle\nabla(\xi\varphi),\nabla(\xi\varphi)\rangle-\langle\varphi,|\nabla\xi|^2\varphi\rangle\,.
\end{align*}
\end{theorem}

\begin{proof}
We start by showing that $\xi\varphi\in H^1$ and $\xi^2\varphi\in H^1$. It is obvious that  $\xi\varphi\in L^2$. To show $\nabla(\xi\varphi)\in L^2$ we use Theorem \ref{gauss}. For $h\in C^\infty_c(\R^d)$ such that $\mathrm{supp}\,h\subseteq B_R(0)$ we write 
\begin{align*}
-\int_{\R^d}(\partial_ih)(\xi\varphi)\textrm dx^d&=-\sum_{j\in\N}\int_{\Omega_j}(\partial_ih)(\xi\varphi)\textrm dx^d\\&=-\sum_{j\in\N}\int_{\Omega_j}\partial_i(h\xi\varphi)-h\partial_i(\xi\varphi)\textrm dx^d\\
&=-\sum_{j\in\N}\int_{\Omega_j}\partial_i(h\xi\varphi)-(h(\partial_i\xi)\varphi+h\xi\partial_i\varphi)\textrm dx^d\\
&=-\sum_{j\in\N}\int_{\Omega_j\cap B_R(0)}\partial_i(h\xi\varphi)-(h(\partial_i\xi)\varphi+h\xi\partial_i\varphi)\textrm dx^d\\
&=-\sum_{j\in\N}\int_{\partial(\Omega_j\cap B_R(0))}h\xi\varphi\nu_i  \mathrm dS^{d-1}+\sum_{j\in\N}\int_{\Omega_j\cap B_R(0)}h(\partial_i\xi)\varphi+h\xi\partial_i\varphi\textrm dx^d\,.
\end{align*}
The expression $-\sum_{j\in\N}\int_{\partial(\Omega_j\cap B_R(0))}h\xi\varphi\nu_i  \mathrm dS^{d-1}$ is $0$ due to the fact that the function $h\xi\varphi$ is continuous and the fact that outer normals $\nu_i$ from two neighboring domains have opposite signs. Therefore we conclude
\begin{align*}
-\int_{\R^d}(\partial_ih)(\xi\varphi)\textrm dx^d&=\sum_{j\in\N}\int_{\Omega_j\cap B_R(0)}h[(\partial_i\xi)\varphi+\xi\partial_i\varphi]\textrm dx^d\\&=\sum_{j\in\N}\int_{\Omega_j}h[(\partial_i\xi)\varphi+\xi\partial_i\varphi]\textrm dx^d\\
&=\int_{\R^d}h[(\partial_i\xi)\varphi+\xi\partial_i\varphi]\textrm dx^d\,.
\end{align*}
Since the left hand side is finite due to
\begin{align*}
\|(\partial_i\xi)\varphi+\xi\partial_i\varphi\|_2\leq\|(\partial_i\xi)\varphi\|_2+\|\xi\partial_i\varphi\|_2\leq
\|\partial_i\xi\|_\infty\|\varphi\|_2+\|\xi\|_\infty\|\partial_i\varphi\|_2<\infty\,.
\end{align*}
This implies $\xi\varphi\in H^1$ and $\nabla(\xi\varphi)=(\nabla\xi)\varphi+\xi\nabla\varphi$. In the same fashion one can check that $\xi^2\varphi\in H^1$. Now we are ready to prove our claim. We write
\begin{align*}
2\mathrm{Re}\langle\nabla(\xi^2\varphi),\nabla\varphi\rangle=&\langle\nabla(\xi^2\varphi),\nabla\varphi\rangle+\langle\nabla\varphi,\nabla(\xi^2\varphi)\rangle\\
=&\langle\xi\nabla(\xi\varphi)+(\nabla\xi)\xi\varphi,\nabla\varphi\rangle+\langle\nabla\varphi,\xi\nabla(\xi\varphi)+(\nabla\xi)\xi\varphi\rangle\\
=&\langle\nabla(\xi\varphi),\nabla(\xi\varphi)-(\nabla\xi)\varphi\rangle+\langle(\nabla\xi)\varphi,\nabla(\xi\varphi)-(\nabla\xi)\varphi\rangle+\\
&\langle\nabla(\xi\varphi)-(\nabla\xi)\varphi,\nabla(\xi\varphi)\rangle+\langle\nabla(\xi\varphi)-(\nabla\xi)\varphi,(\nabla\xi)\varphi\rangle\\
=&2\langle\nabla(\xi\varphi),\nabla(\xi\varphi)\rangle-2\langle(\nabla\xi)\varphi,(\nabla\xi)\varphi\rangle\,
\end{align*}
which completes the proof.
\end{proof}
\end{appendix}

\bibliography{../references}
\bibliographystyle{amsplain}
\end{document}